\documentclass[11pt,a4paper,english]{article}

\pdfoutput=1
\usepackage[T1]{fontenc}
\usepackage[utf8]{inputenc}
\usepackage{babel}
\usepackage{blindtext}

\usepackage{times}
\usepackage{comment}
\usepackage[dvipdfmx]{color}
\usepackage[pdftex]{graphicx}
\def \fsize{10cm}
\def \fext{pdf}
\def \vect#1{\mbox{\boldmath $#1$}}

\newcommand{\rif}{{\rm ~~if~}}

\newcommand{\mc}{\mathcal}

\usepackage{bm,enumerate,amsmath,amssymb,amsthm}
\usepackage{epsfig}
\usepackage{url}
\usepackage{cite}

\newtheorem{definition}{Definition}
\newtheorem{lemma}{Lemma}
\newtheorem{theorem}{Theorem}
\usepackage[title,titletoc]{appendix}

\excludecomment{hdn}
\excludecomment{jap}
{ \begin{hdn} \begin{verbatim} }%
		{ \end{verbatim} \end{hdn} }

\title{Pure Nash Equilibrium and Coordination of Players in Ride Sharing Games}
\author{%
  Tatsuya Iwase \\
  \small Social Systems Research-Domain\\
  \small Toyota Central R\&D Labs., Inc. \\
  \small\texttt{tiwase@mosk.tytlabs.co.jp}
  \and
  Takahiro Shiga\\
  \small Social Systems Research-Domain\\
  \small Toyota Central R\&D Labs., Inc. \\
  \small\texttt{t-shiga@mosk.tytlabs.co.jp}
}

\date{}

\begin{document}
  \maketitle

\begin{abstract}
In this study, we formulate positive and negative externalities caused by changes in the supply of shared vehicles as ride sharing games. The study aims to understand the price of anarchy (PoA) and its improvement via a coordination technique in ride sharing games. A critical question is whether ride sharing games exhibit a pure Nash equilibrium (pNE) since the PoA bound assumes it. Our result shows a sufficient condition for a ride sharing game to have a finite improvement property and a pNE similar to potential games. This is the first step to analyze PoA bound and its improvement by coordination in ride sharing games. We also show an example of coordinating players in ride sharing games using signaling and evaluate the improvement in the PoA.
\end{abstract}
\newpage

\begin{hdn}
	\color{red}
	\begin{verbatim}
シナリオ：
従来CG上でmediateしていた。mediateによる改善幅MRを知る問題。
CGはpNE,FIPという良い性質があり、それを使うとPoA, MRをboundできた。
今回、シェア問題のために、CGを拡張しRSGを作った。PGであればmediate改善幅を知れる。
果たしてRSGはFIPなのか？=gQ1、PoA,MRのboundは幾つになるのか？=gQ2？？
A:あるケースも有る。成り立つ具体的なクラスを示した.oPGになることがわかった=T1
gQ2に対しては今後の課題。ここでは事例だけ示す
今後PoA,MR
	\end{verbatim}
	\color{black}
\end{hdn}

\begin{hdn}
	\color{red}
	\begin{verbatim}
	証明グラフで論旨を確認:
	 thm:driverでH1,H3,H4つかってない
	\end{verbatim}
	\color{black}	

\end{hdn}

\section{Introduction}
\label{sec:intro}

\begin{hdn}
	\color{red}
	\begin{verbatim}
   eG:シェアのためOD協調させたい
  　eA:車初期値をxにしたベイジアンゲーにして、BCEでBNEつくる
  　Q:sharingを効率化するにはどうしたらよいか？
	\end{verbatim}
	\color{black}
\end{hdn}

\begin{jap}
	\color{red}
	\begin{verbatim}
	今後都市化が進むと渋滞やリソース配分が問題になる。使われていないリソースの有効活用は解決策の一つであるが、車の供給変化による外部性は従来のcongestion gameではモデル化できない。本研究では新しくライドシェアリングゲームを定式化し、プレイヤを協調させた時のPoAの改善を検討する。
	\end{verbatim}
	\color{black}
\end{jap}

Congestion and effective resource allocation are traditional problems that game theory has been trying to solve. As populations are increasingly concentrated in big cities, the congestion problem becomes more critical. One practical way to solve the congestion problem is sharing of unused resources. For example, even during heavy traffic congestion, many vehicles have empty seats. Similarly, many empty vehicles occupy limited parking spaces in urban areas while people struggle to find an empty taxi. Besides the transportation field, we can find examples such as unused buildings, empty restaurants, and idle workers.

However, the sharing of unused resources is not realized unless there are matching demands of multiple users. Empty vehicle seats are shared among passengers only if they have a common part on their routes. A vehicle is shared only if the destination of a driver is equal to the origin of another driver. If people behave selfishly, the probability of matching the demand for sharing becomes small, leading to the problem of demand coordination and incentive design.

Particular to the sharing of vehicles is that positive and negative externalities of one player's route change are propagated to other players who do not choose the relevant routes via changes in vehicle supply. Consider a chain of vehicle sharing in which player 1 first drives a vehicle from A to B, then player 2 drives the same vehicle from B to C, and player 3 drives the vehicle from C to D. If player 1 quits to use the vehicle, player 2 and even 3 cannot use the vehicle even though player 3 does not share any part of his route with player 1. Because classical congestion games only focus on externalities among players who choose the same routes, it is necessary to consider alternative games that model externalities via vehicle supply to analyze coordination in vehicle sharing.


This study formulates ride sharing games that model the positive and negative externalities of choosing vehicles. We also consider how to coordinate players to improve the efficiency of sharing, i.e., maximizing the operation rate of otherwise unused vehicles by giving players an incentive. We mainly assume applications in transportation areas, such as carpooling and ride sharing.


\subsection{Related work}
\label{sec:related}

\begin{hdn}
	\color{red}
	\begin{verbatim}
	\end{verbatim}
	\color{black}
\end{hdn}

\begin{jap}
	\color{red}
	\begin{verbatim}
	従来混雑問題はcongestion gameでモデル化され、PoAによる混雑外部性の分析やメカニズムデザイン・シグナリングによる混雑解消案が提案されている。しかしcongestion gameでは、プレイヤA->B->Cと車が乗り継がれた時の、A->Cへの外部性の伝播を表現できない。カーシェアリングのPoAを研究した事例はないとするレビュー報告もある。
	\end{verbatim}
	\color{black}
\end{jap}

Since Rosenthal introduced the congestion game\cite{rosenthal}, it has been applied to problems of congestion externalities in several areas such as transportation and communication networks\cite{cgapp}. In this game, players choose a combination of resources and the players' payoffs depend on the congestion level, i.e., the number of players using the same resources. The congestion game is a subclass of potential games, which feature the {\it finite improvement property} (FIP). This property guarantees that if each player updates his strategy in response to other players by turns, it will reduce his private cost and improve the common potential function, which eventually reaches a local minimum, the {\it pure Nash equilibrium} (pNE), where each player has a deterministic strategy. Because of this property, potential and congestion games have been well studied and have wide applications\cite{sandolm}. 

The negative externalities in congestion games are also well studied. If players choose their routes selfishly in road networks, it causes loss of social welfare compared to socially optimal routing. The ratio of social cost in a selfish choice to the one in the social optimal is called the {\it price of anarchy} (PoA) and its bounds are well known, especially for affine cost functions\cite{poa}. Studies on PoA bound assume that a game has a pNE. Because of the property of always having a pNE, most PoA studies have focused on congestion games.


A major difference between vehicle sharing and traffic routing problems is vehicle supply. While players drive their own vehicles in traditional traffic routing problems, players must find shared vehicles before riding in sharing problems. In the case of a chain of vehicle sharing from player A via player B to player C, the choice of player A has externalities not only on player B but those also propagate to player C, who does not share any common part of route with player A. This kind of complicated externality is not considered in traditional routing problems and congestion games.


According to a review of sharing studies in transportation areas\cite{agatz}, there are few game theory studies on externalities and coordination of players' moves in vehicle sharing problems. Traditional studies on sharing include the optimization of vehicle routes for picking up all passengers (the {\it Dial-a-ride problem}), problems of splitting passengers' fares according to their riding distances, optimization of locations of carpool stations, and problems of relocation of carpool vehicles among stations. Those studies mainly focus on optimization problems under given moves (origins and destinations) of passengers and fixed drivers of vehicles, and therefore, do not analyze the PoA when passengers strategically choose their moves in response to the moves of other passengers and vehicles. The review states that there are no studies on coordination of players' moves to improve the efficiency of sharing or reduce the PoA.

There are several studies that examine the coordination of players in traditional traffic routing and congestion games. One promising technique for coordination is the mechanism design, which mainly provides monetary incentives for players to change their behavior in a coordinated manner. Christodoulou studied coordination mechanism in congestion games\cite{cgmech}.

Another relatively new technique for coordination is signaling, in which a mediator provides information for players to control their beliefs on uncertain environments, and accordingly, the expected payoff and resulting choices when there is information asymmetry between the mediator and players\cite{kame,bce,kremer}. Most recent studies are based on the {\it revelation principle}, which proved the existence of incentive compatible recommendations of choices equivalent to the raw information inducing the same choices\cite{rev}. Rogers applied differential privacy techniques originally from database security to traffic routing problems for mitigating congestions by sending noisy incentive compatible recommendations of routes\cite{rogers}. Vasserman also applied recommendations to traffic routing and analyzed how the PoA improved\cite{vasserman}.

\subsection{Our contributions and paper structure}
\label{sec:contribution}

\begin{hdn}
	\color{red}
	\begin{verbatim}
	G:シェアゲーのcoordinate効果を評価したい。CGと同じくPoAで。
	Q:シェアゲーはFIPか？pNEあるのか？
	A:CGにないこういう要素をformulateした。pNEある。T1,2,3を英語で
	S:section。coordinateは事例で
	
	シェアでゲー理論でPoAっていう研究はまだない
	これまでとは違うポテンシャル関数
	\end{verbatim}
	\color{black}
\end{hdn}

\begin{jap}
	\color{red}
	\begin{verbatim}
	車のシェアリングにおける外部性をライドシェアリングゲームとしてモデル化し、シグナリングによる協調効果をPoAによって評価することを目的とする。PoAは純粋ナッシュ均衡を前提としているので、本研究ではライドシェアリングゲームが純粋ナッシュ均衡を持つ十分条件を示した。またシグナリングによるPoA改善効果を事例評価した。
	\end{verbatim}
	\color{black}
\end{jap}

This study proposes ride sharing games as a formulation of positive and negative externalities caused by changes in the supply of shared vehicles. The study's objective is to understand the PoA and its improvement via a coordination technique in ride sharing games. A critical question is whether ride sharing games have a pNE since the PoA bound assumes it. Our result shows a sufficient condition for a ride sharing game to have a FIP and a pNE similar to potential games. This is the first step to analyze PoA bound and its improvement by coordination in ride sharing games. We also show an example of coordinating players in ride sharing games using signaling and evaluate the improvement in the PoA.


Section \ref{sec:model} provides a formulation of ride sharing games and other setups. Section \ref{sec:theorem} presents the main results on a condition of ride sharing games to have pNE and its proof. Section \ref{sec:exam} presents graphic examples of ride sharing games and coordination of players by signaling.

\section{The Models}
\label{sec:model}

\begin{jap}
	\color{red}
	\begin{verbatim}
以下の各種数学的定義。
・ライドシェアリングゲーム
・ポテンシャルゲーム関連の用語
・ベイズライドシェアリングゲーム
・シグナリング関連の用語
	\end{verbatim}
	\color{black}
\end{jap}

\subsection{Ride sharing games}
\label{sec:rsg}

\begin{hdn}
	\color{red}
	\begin{verbatim}
	\end{verbatim}
	\color{black}
\end{hdn}

A {\it ride sharing game} is defined as a tuple $G=<\mc{N},\mc{M},\mc{T},\mc{G},\mc{A},\mu,c>$, where 
\begin{itemize}
	\item $\mc{N}=\{1,\ldots,N\}$ is a finite set of players. A player $i \in \mc{N}$ represents an user of shared vehicles. $-i$ represents all players except for $i$.
	\item $\mc{M}=\{1,\ldots,M\}$ is a finite set of vehicles. Each vehicle $m \in \mc{M}$ has a common seating capacity $w \in \mathbb{N}_{> 0}$.
	\item $\mc{G}=<\mc{V},\mc{E}>$ is a directed graph that has a finite set of nodes $\mc{V}=\{1,\ldots,V\}$ and a finite set of edges $\mc{E}=\{1,\ldots,E\}$. $\mc{G}$ is a simple graph but each node has a loop to itself. A node $v \in \mc{V}$ represents a place and an edge $e \in \mc{E}$ represents a road. Players and vehicles move on $\mc{G}$.
	\item $\mc{T}=\{1,\ldots,T\}$ is a finite set of time that partitions the day. Each player and vehicle is located on a node at time $t \in \mc{T}$ and finishes a move on an edge during period $(t,t+1)$.
	\item $\mc{R}$ is a set of all paths with length $T-1$ on $\mc{G}$. A path $r=(v_{1},e_{1},v_{2},e_{2},\ldots,e_{T-1},v_{T}) \in \mc{R}$ represents a round trip of a player on a day. We denote $r_{k} \subseteq r$ if $r_{k}$ is a induced path of $r$. $r-r_{k}$ is a complement of $r_{k}$ which is also a induced path of $r$ including all edges not in $r_{k}$. $r$ is a common path $r=r_{1} \bigcap r_{2}$ if and only if $r \subseteq r_{1}$ and $r \subseteq r_{2}$.
	
	\item $\mc{A}_{i} \subset \mc{R}$ is a set of strategies of player $i$. $\mc{A}=\underset{i \in \mc{N}}{\times}\mc{A}_{i}$ is a set of strategy profiles. $a_{i} \in \mc{A}_{i}$ is a round trip of player $i$ and $\vect{a} \in \mc{A}$ is a strategy profile. $\vect{a}_{-i}$ represents a strategy profile of all players except for $i$. A strategy update is denoted as $(\vect{a},b_{i},i)$ when a original strategy profile is $\vect{a}$ and player $i$ update a strategy from $a_{i}$ to $b_{i}$.
	\item $\mu(i,t,\vect{a}):<\mc{N},\mc{T},\mc{A}> \to \mc{M}$ is a map that represents the allocation of player $i$ to vehicle $m$ during each period $(t,t+1)$ depending on strategy profile $\vect{a}$. In the case where no vehicle is allocated to player $i$, $\mu(i,t,\vect{a})=\emptyset$. Each vehicle $m$ moves together with allocated player $i$ on the same edge where the player moves. $s_{m}(t,\vect{a})$ represents the number of players riding on vehicle $m$ during period $(t,t+1)$ when the strategy profile is $\vect{a}$.
	\item $c_{e}(w,s_{m}):<\mathbb{N}_{> 0},\mathbb{N}_{\ge 0}> \to \mathbb{R}_{\ge 0}$ is a cost function of a player riding on vehicle $m$ on edge $e$. $c=\{c_{e}|e \in \mc{E}\}$ is a set of cost functions of all edges. The total cost of player $i$ on a day is $c_{i}(\vect{a})=\sum_{e_{t} \in a_{i}}c_{e}(w,s_{\mu(i,t,\vect{a})}(t,\vect{a}))$.
\end{itemize}

In this study, we consider one-shot games, where players simultaneously choose whole round trips $\vect{a}$ on the day. We assume that cost function $c_{e}$ is monotone decreasing for $s_{m}$ when $s_{m} < w$ and monotone increasing when $s_{m} \ge w$.

\subsection{Basic notions of potential games}
\label{sec:fip}

\begin{hdn}
	\color{red}
	\begin{verbatim}
	\end{verbatim}
	\color{black}
\end{hdn}

Here we prepare the basic concepts of potential games used in the rest of the paper. A {\it pure Nash equilibrium} (pNE) of game $G=<\mc{N},\mc{A},c>$ is defined as a deterministic strategy profile $\vect{a} \in \mc{A}$ if and only if no player can reduce his cost by updating his deterministic strategy from $\vect{a}$.

While not all games have a pNE, some games always have a pNE. In particular, the following property guarantees the existence of a pNE.

\begin{definition}[Finite improvement property]
	A game has a finite improvement property (FIP) if a strategy profile of the game always converges to a pure Nash equilibrium by updating each player's strategy, by turns, finite times.
	\label{def:fip}
\end{definition}

When a game has an FIP, a pNE can be found by players updating their strategies one player at a step. Therefore, it is unnecessary to search the whole set of strategy profiles and the computational cost to find a pNE is reduced.

A potential game is game $G$ which has a potential function $\Phi$ defined as follows\cite{opg}.

\begin{definition}[(Ordinal) potential game]
	$\Phi:\mc{A} \to \mathbb{R}$ is a {\it (ordinal) potential function} and $G$ is a {\it (ordinal) potential game} if
	\begin{equation}
	sgn(\Phi(b_{i},\vect{a}_{-i})-\Phi(a_{i},\vect{a}_{-i}))=sgn(c_{i}(b_{i},\vect{a}_{-i})-c_{i}(a_{i},\vect{a}_{-i})).
	\label{eq:pg}
	\end{equation}
	\label{def:pg}
\end{definition}

The following theorem guarantees that potential games have an FIP and then a pNE.

\begin{theorem}
	If a potential game has a potential function with a finite amount of values, the game has an FIP.
	\label{thm:pg}
\end{theorem}

\begin{proof}
	Starting from an arbitrary strategy profile $\vect{a}$, each player updates their strategy by turns to minimize his cost. From the definition, it also reduces the value of the potential function until it reaches a local minimum. At that point, no player can reduce his cost and it is a pNE. From Definition \ref{def:fip}, it also has an FIP.
\end{proof}

In applications such as transportation and communication networks, problems are often formulated as a minimization of a social cost, which is the total cost of all players. Normally, a social cost in a pNE is not the same as in the optimal. The ratio between the cost of a pNE and an optimal cost is called the {\it price of anarchy} (PoA) and is defined as follows\cite{poa}.

\begin{definition}[Price of anarchy]
	The price of anarchy $\rho(G)$ of game $G$ is
	\begin{equation}
	\rho(G)=\frac{C(\vect{a}^*)}{C(\vect{a}^{opt})}
	\label{eq:poa}
	\end{equation}
	where $C(\vect{a})$ is a social cost of strategy profile $\vect{a}$, $\vect{a}^*$ is pNE and $\vect{a}^{opt}$ is the social optimal.
	\label{def:poa}
\end{definition}

Since the PoA assume a pNE, former studies on the PoA have focused on games with a pNE, such as congestion games.

\subsection{Bayesian ride sharing games}
\label{sec:brsg}

\begin{hdn}
	\color{red}
	\begin{verbatim}
	\end{verbatim}
	\color{black}
\end{hdn}

We consider cases where players have incomplete information on vehicle allocations. A {\it Bayesian ride sharing game} is an extension of a ride sharing game and defined as $G_{b}=<\mc{N},\mc{M},\mc{T},\mc{G},\mc{A},\mc{X},\mc{P},\mu,c>$ where
\begin{itemize}
	\item $\mc{X}$ is a set of possible values of an exogenous variable $x \in \mc{X}$, which affects the allocation of vehicles $\mu$.
	\item $\mu(i,t,\vect{a}|x)$ is the allocation of vehicles depending on $x$. Similarly, $s_{m}(t,\vect{a}|x)$ is the number of players on vehicle $m$ depending on $x$ and $c_{i}(\vect{a}|x)$ is the cost of player $i$ depending on $x$. 
	\item $p_{i}(x):\mc{X} \to [0,1]$ is a probability distribution on $X$ of player $i$, which represents his belief. $\mc{P}=\{p_{i}|i \in N\}$ is a set of probability distributions of all players.
	\item Definitions of other elements of $G_{b}$ are the same as those of ride sharing game $G$.
\end{itemize}\

Examples of exogenous variable $x$ are initial vehicle locations and traffic accidents. Each player chooses $a_{i}$ to minimize his expected cost, which is $\mathbb{E}_{\mc{X}}[c_{i}]=\sum_{x \in \mc{X}}c_{i}(\vect{a}|x)p_{i}(x)$.

A pure Bayesian Nash equilibruim (pBNE) of a bayesian game is a similar concept to pNE of a deterministic game that is defined as a deterministic strategy profile $\vect{a} \in \mc{A}$ if and only if no player can reduce his expected cost by updating his deterministic strategy from $\vect{a}$. 

\subsection{Signaling on Bayesian ride sharing games}
\label{sec:signal}

\begin{hdn}
	\color{red}
	\begin{verbatim}
	\end{verbatim}
	\color{black}
\end{hdn}

In this study, we use the {\it Bayes correlated equilibrium} (BCE)\cite{bce} as a signaling technique of a mediator to coordinate players. A BCE is a conditional distribution $\sigma(\hat{\vect{a}}|x)$ of a random recommendation $\hat{\vect{a}}$, which is {\it incentive compatible} (IC) as defined below.

\begin{definition}[Incentive compatible]
	A recommendation policy $\sigma(\hat{\vect{a}}|x)$ is incentive compatible if
	\begin{equation}
	\sum_{\hat{\vect{a}}_{-i},x}p_{i}(x)\sigma(\hat{\vect{a}}|x)c_{i}(\hat{a}_{i},\hat{\vect{a}}_{-i}|x) \leq \sum_{\hat{\vect{a}}_{-i},x}p_{i}(x)\sigma(\hat{\vect{a}}|x)c_{i}(a_{i},\hat{\vect{a}}_{-i}|x),\: \forall i \forall \hat{a}_{i}.
	\label{eq:ic}
	\end{equation}
	\label{def:ic}
\end{definition}

Given the cost function of the mediator $c_{s}(\vect{a}|x)$, the problem of the mediator is to design an optimal IC recommendation that makes players coordinate to minimize their cost. The problem is expressed as follows.

\begin{equation}
\left.
\begin{array}{l}
\max_{\sigma} \mathbb{E}_{x}[c_{s}(\hat{\vect{a}}|x)] \\
s.t. \: \sum_{\hat{\vect{a}}_{-i},x}p_{i}(x)\sigma(\hat{\vect{a}}|x)c_{i}(\hat{a}_{i},\hat{\vect{a}}_{-i}|x) \leq \sum_{\hat{\vect{a}}_{-i},x}p_{i}(x)\sigma(\hat{\vect{a}}|x)c_{i}(a_{i},\hat{\vect{a}}_{-i}|x),\: \forall i \forall \hat{a}_{i}.
\end{array}
\right.
\label{eq:problem}
\end{equation}

\section{Result}
\label{sec:theorem}

\begin{hdn}
	\color{red}
	\begin{verbatim}
	PoAに関するT
	\end{verbatim}
	\color{black}

\end{hdn}

\begin{jap}
	\color{red}
	\begin{verbatim}
ライドシェアリングゲームが純粋ナッシュ均衡を持つ十分条件とその証明
eq.15のようなポテンシャル関数がある
	\end{verbatim}
	\color{black}
\end{jap}

Here, we discuss when ride sharing games have a pNE in order to evaluate the PoA of a game and its improvement by coordination.

We first start with the following negative result in most general cases.

\begin{theorem}
	There exist ride sharing games that do not have an FIP.
	\label{thm:nonpg}
\end{theorem}

\begin{proof}
	The example in Section \ref{sec:nonpg} shows the case in which the strategy updates of players are caught in an infinite loop, which will not converge to any pNE.
\end{proof}

If all ride sharing games fall into this case, it is hard to apply theory of PoA. However, we found cases where ride sharing games have an FIP. Intuitively, cost functions $c_{e}$ become monotone decreasing when $N \le w$ and then ride sharing games have a structure of {\it increasing returns} and are {\it locked in} to a pNE\cite{arthur}. Before proceeding, we introduce several notions here.

Let $M_{et}(\vect{a})$ and $N_{et}(\vect{a})$ be the number of vehicles and players on edge $e$ during period $(t,t+1)$, respectively. The change in $N_{et}$ by a strategy update $(\vect{a},b_{i},i)$ is defined as follows.
 
\begin{equation}
	\Delta N_{et}(\vect{a},b_{i},i)=N_{et}(b_{i},\vect{a}_{-i})-N_{et}(\vect{a}).
	\label{eq:neinc}
\end{equation}

\begin{definition}[No-vehicle-loss update]
	A strategy update $(\vect{a},b_{i},i)$ is no-vehicle-loss if $M_{et}(b_{i},\vect{a}_{-i}) \not= 0, \forall \{e_{t}: e_{t} \in b_{i}, M_{et}(\vect{a}) > 0\}$.
	\label{def:nondecu}
\end{definition}
 
An allocation map $\mu$ of a ride sharing game can be divided into {\it path allocation} $\mu_{r}$ and {\it seat allocation} $\mu_{s}$. Path allocation determines the paths of vehicles. Once an edge on which a vehicle moves has been fixed, seat allocation determines the allocation of players to vehicles on the edge. The following path allocation assumes that the more demands there are on an edge, the more vehicles are allocated to it.

\begin{definition}[Linear path allocation]
	A linear path allocation $\mu_{r}$ determines an allocation of $M_{et}$ out of $M$ vehicles on a node $v$ to outgoing edge $e \in \mc{E}_{v}$ on which $N_{et}$ players move at $t$ so that
	\begin{equation}
	M_{et} = floor(k*N_{et})
	\label{eq:linpath}
	\end{equation}
	where $k$ is a constant that keeps $M=\sum_{e \in \mc{E}_{v}}k*N_{et}$. The remaining $M-\sum_{e_{t} \in \mc{E}_{v}}M_{et}$ vehicles are allocated to $\mc{E}_{v}$ in the order of $k*N{et}-M_{et}$.
	\label{def:linpath}
\end{definition}

\begin{definition}[Allocated path]
	$r$ is an allocated path if at least one vehicle is allocated on all edges in $r$.
	\label{def:allocp}
\end{definition}

Seat allocation is a simple version of bin packing problem. If players are willing to share a vehicle to reduce their costs, it is natural to assume the first-fit algorithm\cite{binpack} as follows. 

\begin{definition}[First-fit seat allocation]
	A seat allocation is first-fit if players are allocated to a vehicle with the smallest $m$ on the edge until it becomes full.
	\label{def:firstfit}
\end{definition}

This definition immediately yields the following lemma.

\begin{lemma}
	If $N \leq w$ and $\mu_{s}$ is the first-fit seat allocation, all $N_{et}$ players ride in the same vehicle if $M_{et} > 0$.
	\label{thm:acar}
\end{lemma}

\begin{definition}[First fit linear allocation]
	An allocation $\mu$ is first-fit liner if it comprises a linear path allocation and a first-fit seat allocation.
	\label{def:maxnondec}
\end{definition}

The following lemma states that copying the strategy of another player $j$ always results in a cost less than $c_{j}$ in some ride sharing games. 

\begin{lemma}
	$c_{i}(a_{j}, \vect{a}_{-i}) \le c_{j}(\vect{a})$ if
	\begin{itemize}
		\item[H1:] all players have a common set of actions $\mc{A}_{0}$,
		\item[H2:] $N \leq w$,
		\item[H3:] $\mu_{s}$ is the first-fit seat allocation, and
		\item[H4:] strategy update $(\vect{a},a_{j},i)$ is no-vehicle-loss.
	\end{itemize}
	\label{thm:rspg}
\end{lemma}

\begin{proof}
	Let $\vect{a}$ be a current strategy profile. Now consider strategy update $(\vect{a},a_{j},i)$ which is always possible because of H1. Since all $N_{et}$ players ride on a vehicle on the edge according to H2, H3 and Lemma\ref{thm:acar}, the number of players sharing a car $s_{m}$ now depends only on $M_{et}$ and $N_{et}$ as follows.
	\begin{equation}
	s_{m}(t,\vect{a})=s_{et}(\vect{a})=
	\left\{
	\begin{array}{ll}
	N_{et}(\vect{a}) & \rif M_{et}(\vect{a}) \ge 1\\
	0 & \rif M_{et}(\vect{a})=0.
	\end{array}
	\right.
	\label{eq:se}
	\end{equation}
	Since player $i$ joins in $a_{j}$ and all other players' strategies remain the same,
	
	\begin{equation}
	\Delta N_{et}(\vect{a},a_{j},i) \ge 0, \forall e_{t} \in a_{j}.
	\label{eq:neinc}
	\end{equation}
	Therefore, from H4 and Eq.\ref{eq:se}, we get 
	\begin{equation}
	\Delta s_{et}=s_{et}(a_{j},\vect{a}_{-i})-s_{et}(\vect{a}) \ge 0, \forall e_{t} \in a_{j}.
	\label{eq:seinc}
	\end{equation}
	From H2, cost functions $c_{e}$ are monotone decreasing such that
	\begin{equation}
	c_{e}(w,s_{et}(a_{j},\vect{a}_{-i})) \le c_{e}(w,s_{et}(\vect{a})), \forall e_{t} \in a_{j}.
	\label{eq:cedec}
	\end{equation}
	Then,
	\begin{equation}
	c_{i}(a_{j},\vect{a}_{-i}) = \sum_{e_{t} \in a_{j}}c_{e}(w,s_{et}(a_{j},\vect{a}_{-i})) \le \sum_{e_{t} \in a_{j}}c_{e}(w,s_{et}(\vect{a})) = c_{j}(\vect{a}).
	\label{eq:cij}
	\end{equation}
	This completes the proof of Lemma \ref{thm:rspg}.
\end{proof}

\begin{definition}[Riding path]
	$r$ is a riding path of player $i$ if $r \subseteq a_{i}$ and $r$ is an allocated path.
	\label{def:ridingp}
\end{definition}

\begin{definition}[Necessary path]
	$r \in \mc{R}_{n}$ is a necessary path if
	\begin{equation}
	c_{i}(a_{i} \supseteq r, \vect{a}_{-i}) \le c_{i}(b_{i} \not\supseteq r, \vect{a}_{-i}), \forall i, \vect{a}_{-i}
	\label{eq:necp}
	\end{equation}
	\label{def:necp}
	when $r$ is an allocated path.
\end{definition}

\begin{definition}[Sufficient path]
	$r \in \mc{R}_{s}$ is a sufficient path if $r_{c}$ is a riding path of player $i$ and
	\begin{equation}
	c_{i}(r_{c} = r, \vect{a}_{-i}) \le c_{i}(r_{c} \not = r, \vect{a}_{-i}), \forall i, \vect{a}_{-i}.
	\label{eq:sufp}
	\end{equation}
	\label{def:sufp}
\end{definition}

\begin{definition}[Disjoint path set]
	A set of paths $\mc{R}$ is disjoint if
	\begin{equation}
	r_{i} \bigcap r_{j} = \emptyset, \forall \{i \not= j:r_{i},r_{j} \in \mc{R} \}.
	\label{eq:disjp}
	\end{equation}
	\label{def:disjp}
\end{definition}

\begin{definition}[Driver and passener]
	Let $r$ be a necessary path and $r_{c} \supseteq r$ is a riding path of player $i$. $i$ is a driver if $r_{c} \supset r$. $i$ is a passenger if $r_{c} = r$.
	\label{def:driver}
\end{definition}

\begin{lemma}
	Once a player update his strategy and becomes a driver, he cannot improve his cost by updating his strategy again if
	\begin{itemize}
		\item[H1:] all players have a common set of actions $\mc{A}_{0}$,
		\item[H2:] $M=1$,
		\item[H3:] $N \leq w$,
		\item[H4:] $\mu_{s}$ is the first-fit seat linear allocation, and
		\item[H5:] $G$ has a disjoint set of necessary and sufficient paths $\mc{R}_{ns}$.
	\end{itemize}
	\label{thm:driver}
\end{lemma}


\begin{proof}
	Let $i$ be a driver who has a riding path $r_{c}$ including a necessary and sufficient path $r$. From H2, H5 and Definitions \ref{def:sufp} and \ref{def:driver}, all other players update their strategies to be passengers of $r$. Then the number of players on $r$ increases compaired to the one when $i$ lastly updated his strategy. Since $\mc{R}_{ns}$ is disjoint, the number of players on the other $r_{ns} \in \mc{R}_{ns}$ decreases. Then the cost for the part $r$ stays minimal since cost functions are monotone decreasing.
	Meanwhile, $i$ is a only player who is on $r_{c}-r$ since all other players are passengers of $r$. Then the cost for the part $a_{i}-r$ is independent of $\vect{a}_{-i}$ and stay minimal since when $i$ updated strategy last time.
	Therefore, the cost of whole path $a_{i}$ is minimal and then player $i$ cannot update his strategy.
\end{proof}

The following theorem tells us there is a class of ride sharing games that has an FIP.

\begin{theorem}
	A ride sharing game $G$ has an FIP if
	\begin{itemize}
		\item[H1:] all players have a common set of actions $\mc{A}_{0}$,
		\item[H2:] $M=1$,
		\item[H3:] $N \leq w$,
		\item[H4:] $\mu_{s}$ is the first-fit seat linear allocation, and
		\item[H5:] $G$ has a disjoint set of necessary and sufficient paths $\mc{R}_{ns}$.
	\end{itemize}
	\label{thm:rsone}
\end{theorem}

\begin{proof}
	Let $\Phi$ be a function defined as
	\begin{equation}
	\Phi(\vect{a})=\min_{k} c_{k}(\vect{a}).
	\label{eq:rspf}
	\end{equation}
	Let $\vect{a}$ be the current strategy profile, let $j$ be a player who has the current minimum cost among all players, and let his current profile be $a_{j}$. Then, $\Phi$ becomes
	\begin{equation}
	\Phi(\vect{a})=c_{j}(\vect{a}).
	\label{eq:cj}
	\end{equation}

	If $j$ has no riding path, no player are allocated a vehicle because of Definition \ref{def:necp}. In this case, player $i$ can update his strategy to $b_{i}$, which has a necessary allocated path. From Definition \ref{def:necp}, $b_{i}$ updates the minimum cost and accordingly $\Phi(b_{i},\vect{a}_{-i}) \le \Phi(\vect{a})$.
	
		
	If $j$ has a riding path $r_{j}$, it must include a necessary path $r$ because of Definition \ref{def:necp}. If $i$ is not a driver, he can copy the strategy of $j$ as $b_{i}=a_{j}$ without reducing $N_{et}$ for all $e_{t} \in r_{j}$ and from H4 the strategy update $(\vect{a},a_{j},i)$ is no-vehicle-loss. Then from Lemma \ref{thm:rspg}, it updates the minimum cost and accordingly $\Phi(b_{i},\vect{a}_{-i}) \le \Phi(\vect{a})$.
		
	

	If $j$ has a riding path $r_{j}$ and $i$ is a driver, $i$ cannot update his strategy according to Lemma \ref{thm:driver}.

	Accordingly, if a player can update his strategy and reduce cost, it also reduces $\Phi(\vect{a})$ or a player cannot update his strategy. Then, $\Phi(\vect{a})$ satisfies Definition \ref{def:pg} and is a potential function. Consequently, from Theorem \ref{thm:pg}, game $G$ has an FIP. This complete the proof of Theorem \ref{thm:rsone}.

\end{proof}

Although Theorem \ref{thm:rsone} has several assumptions and covers only a limited class of ride sharing games, the following theorem indicates a possibility of relaxation of the assumptions.

\begin{theorem}
	The assumptions in Theorem \ref{thm:rsone} are not necessary conditions.
	\label{thm:rshope}
\end{theorem}

\begin{proof}
	The example in Section \ref{sec:pg2} shows the case where a game has an FIP even though it does not satisfy all the assumptions in Theorem \ref{thm:rsone}.
\end{proof}

\section{Examples}
\label{sec:exam}

\begin{hdn}
	\color{red}
	\begin{verbatim}
	\end{verbatim}
	\color{black}
\end{hdn}

\begin{jap}
	\color{red}
	\begin{verbatim}
	事例
	・純粋ナッシュ均衡がない場合
	・純粋ナッシュ均衡がある場合(証明済み)
	・純粋ナッシュ均衡がある場合(未証明のケース)
	・シグナリングによるPoA改善事例
	\end{verbatim}
	\color{black}
\end{jap}

In this section, we provide graphical examples of ride sharing games. In Sections \ref{sec:nonpg} to \ref{sec:pg2} we assume the following games.
\begin{itemize}
	\item $N=3, T=5, V=4$.
	\item $\mc{G}$ is a complete graph but each node has a loop edge connected to itself to represent staying of players.
	\item Initial location of players is node 1 and that of vehicles is node 2.
	\item $a_{i}$ must include nodes 3 and 4 for all players.
	\item $\mu$ is the first-fit linear allocation.
\end{itemize} 

All edges have the same cost function, which is
\begin{equation}
 c_{e}=
\left\{
\begin{array}{ll}
\frac{d}{s_{m}+1} & \rif s_{m} \le w \\
d(1-\frac{w^2}{s_{m}(w+1)}) & \rif s_{m} > w
\end{array}
\right.
\label{eq:cost}
\end{equation}
where $d=0$ for loop edges and $d=1$ for others.

\subsection{A case of non-FIP game}
\label{sec:nonpg}

\begin{hdn}
	\color{red}
	\begin{verbatim}
	\end{verbatim}
	\color{black}
\end{hdn}

Here, we assume $M=2$ and $w=1$ and then $G$ does not have an FIP. Figure \ref{fig:nonpg0} shows the initial state of this game. The numbers represent nodes $v$; $a,b$ and $c$ represent the players, and $*$ represents a vehicle. Figures \ref{fig:nonpg1} and \ref{fig:nonpg2} show how strategy updates make a loop and the FIP is broken. In Figure \ref{fig:nonpg1}, there are two drivers ($a$ and $c$) and $b$ is the player with the minimum cost. However, in Figure \ref{fig:nonpg2}, $a$ updates his strategy to quit being a driver and become a passenger to reduce his cost, and $b$ loses a vehicle and his cost increases. In this case, $w=1<N$ so that player $b$ cannot reduce cost by copying another player's strategy according to Lemma \ref{thm:rspg}. Then, $b$ must choose the other vehicle to reduce his cost again as in Figure \ref{fig:nonpg3}. This negative externality makes an infinite loop of this driver switching behavior and the game loses its FIP. 

\begin{figure}[ht]
	\centering
	\includegraphics[clip,width=\fsize]{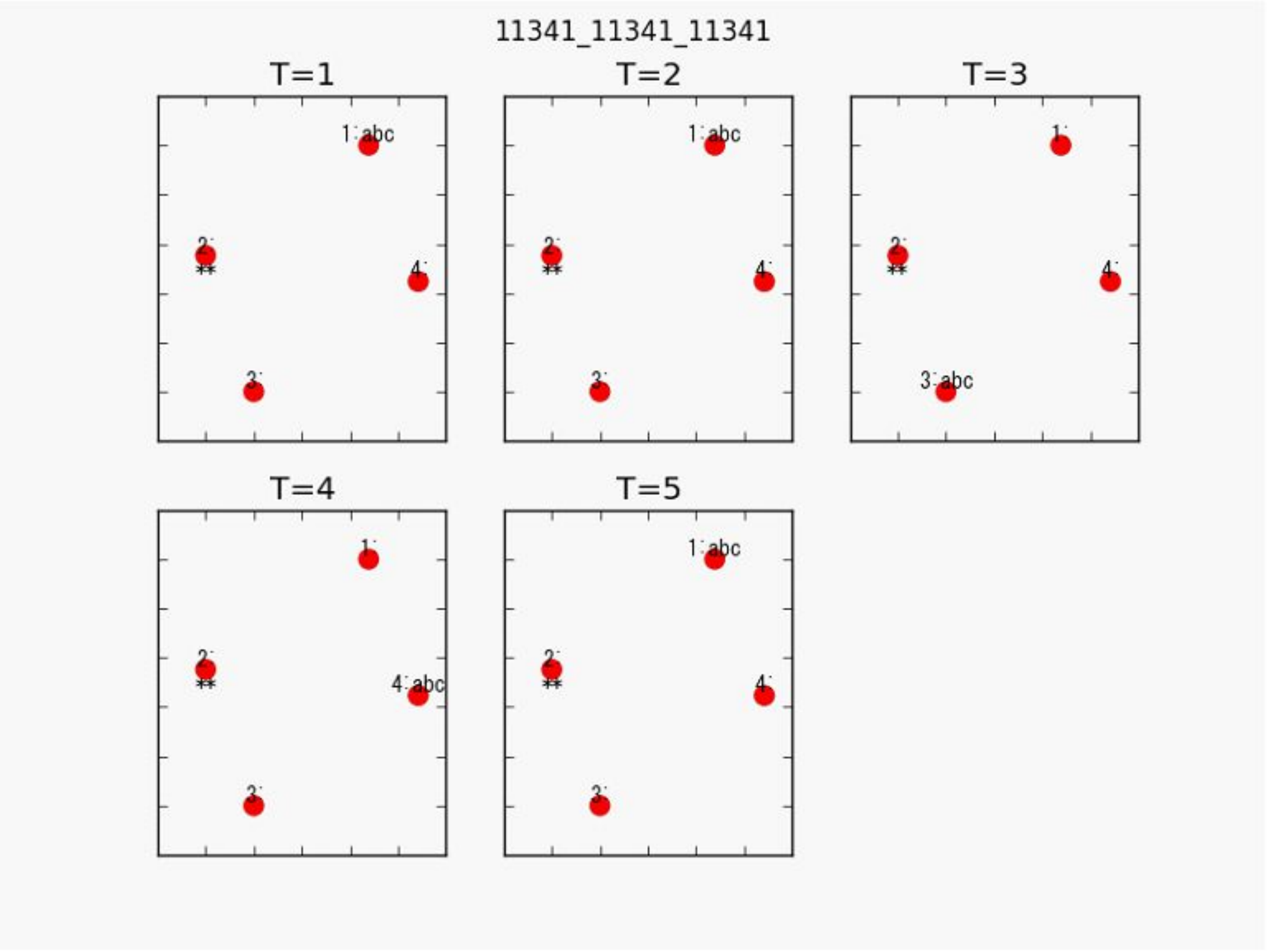}
	\caption{Initial strategy profile of non-FIP game}
	\label{fig:nonpg0}	
\end{figure}

\begin{figure}[ht]
	\centering
	\includegraphics[clip,width=\fsize]{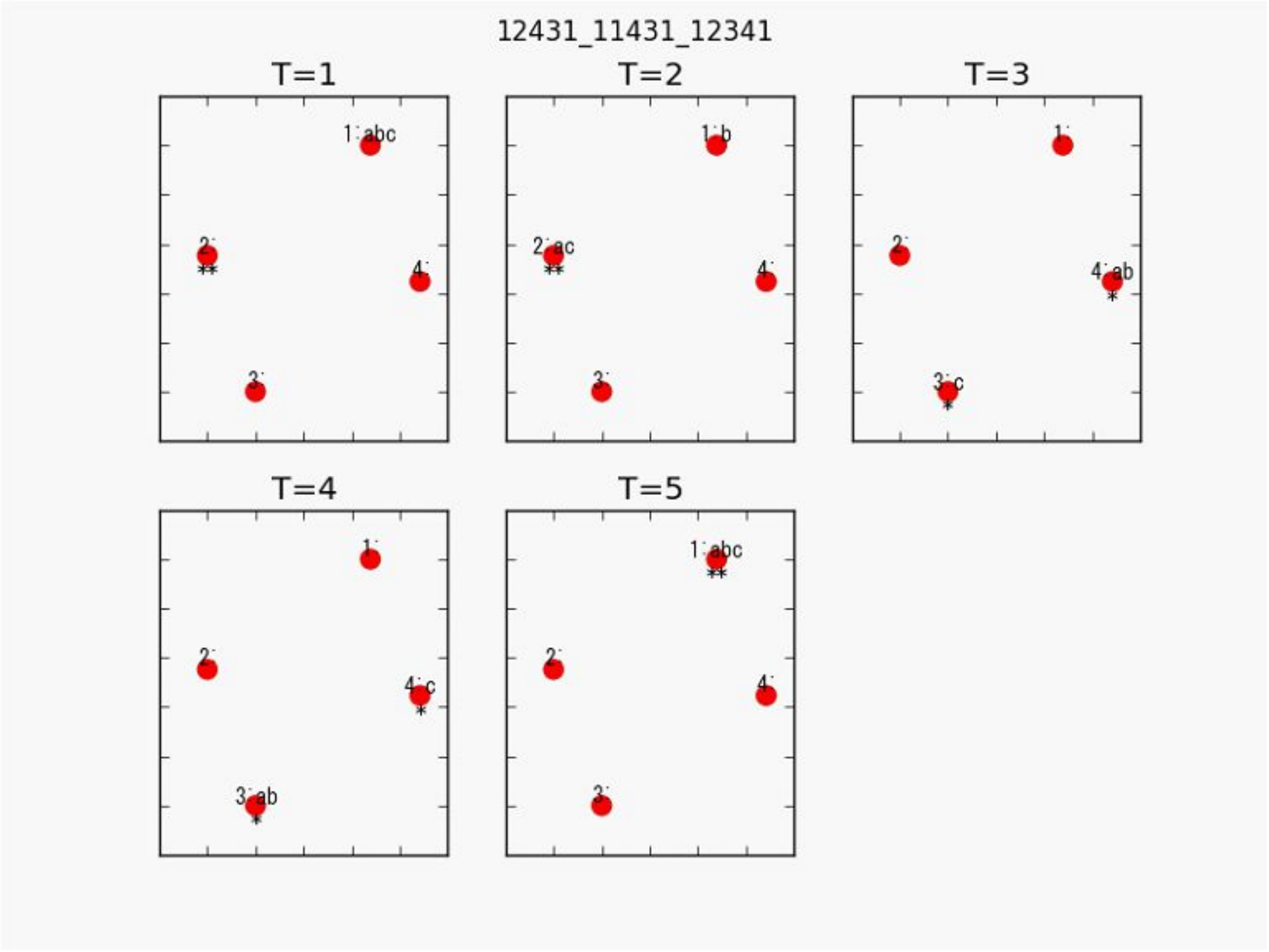}
	\caption{Two drivers}
	\label{fig:nonpg1}
\end{figure}

\begin{figure}[ht]
	\centering
	\includegraphics[clip,width=\fsize]{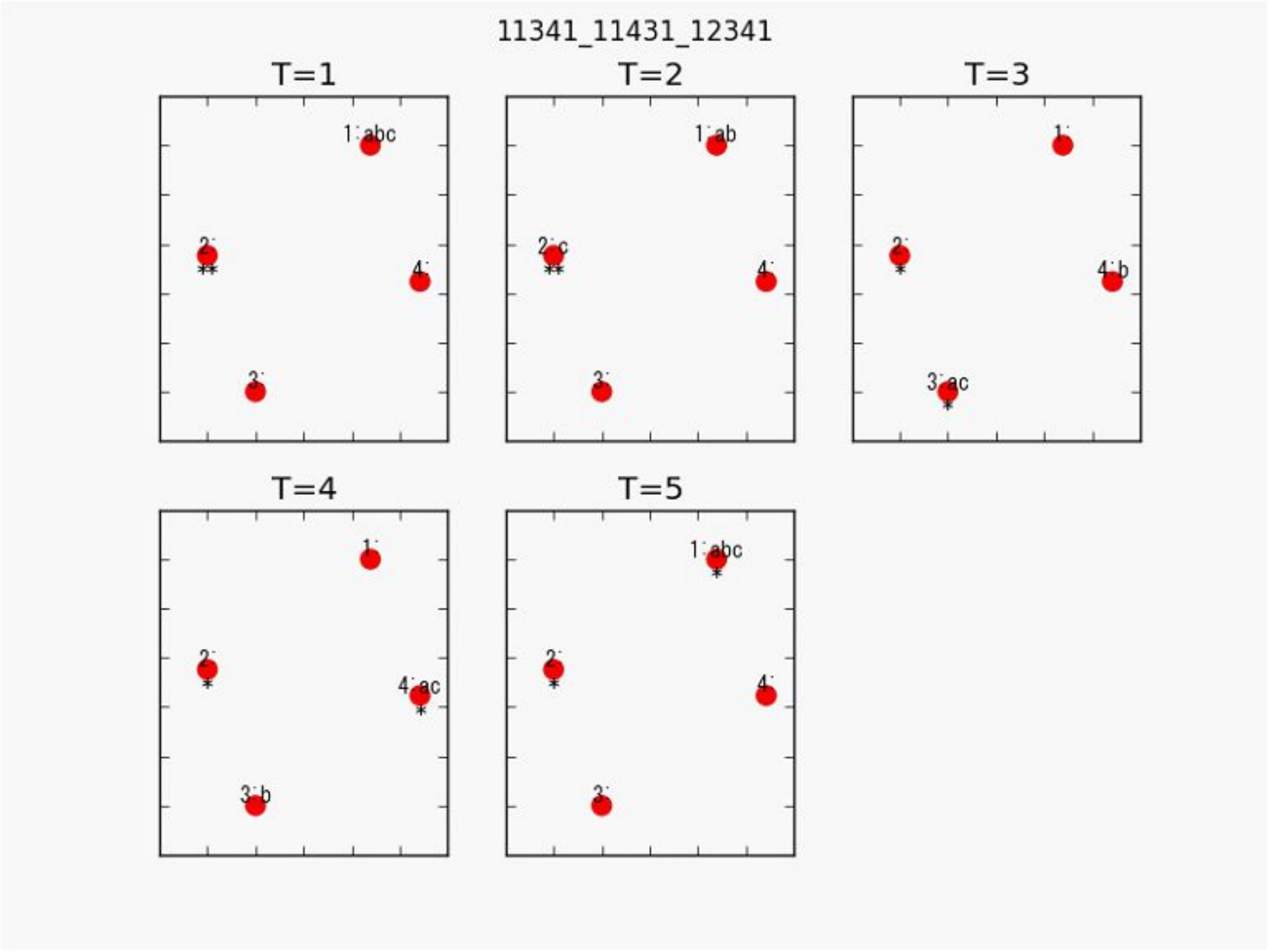}
	\caption{Negative externality}
	\label{fig:nonpg2}	
\end{figure}

\begin{figure}[ht]
	\centering
	\includegraphics[clip,width=\fsize]{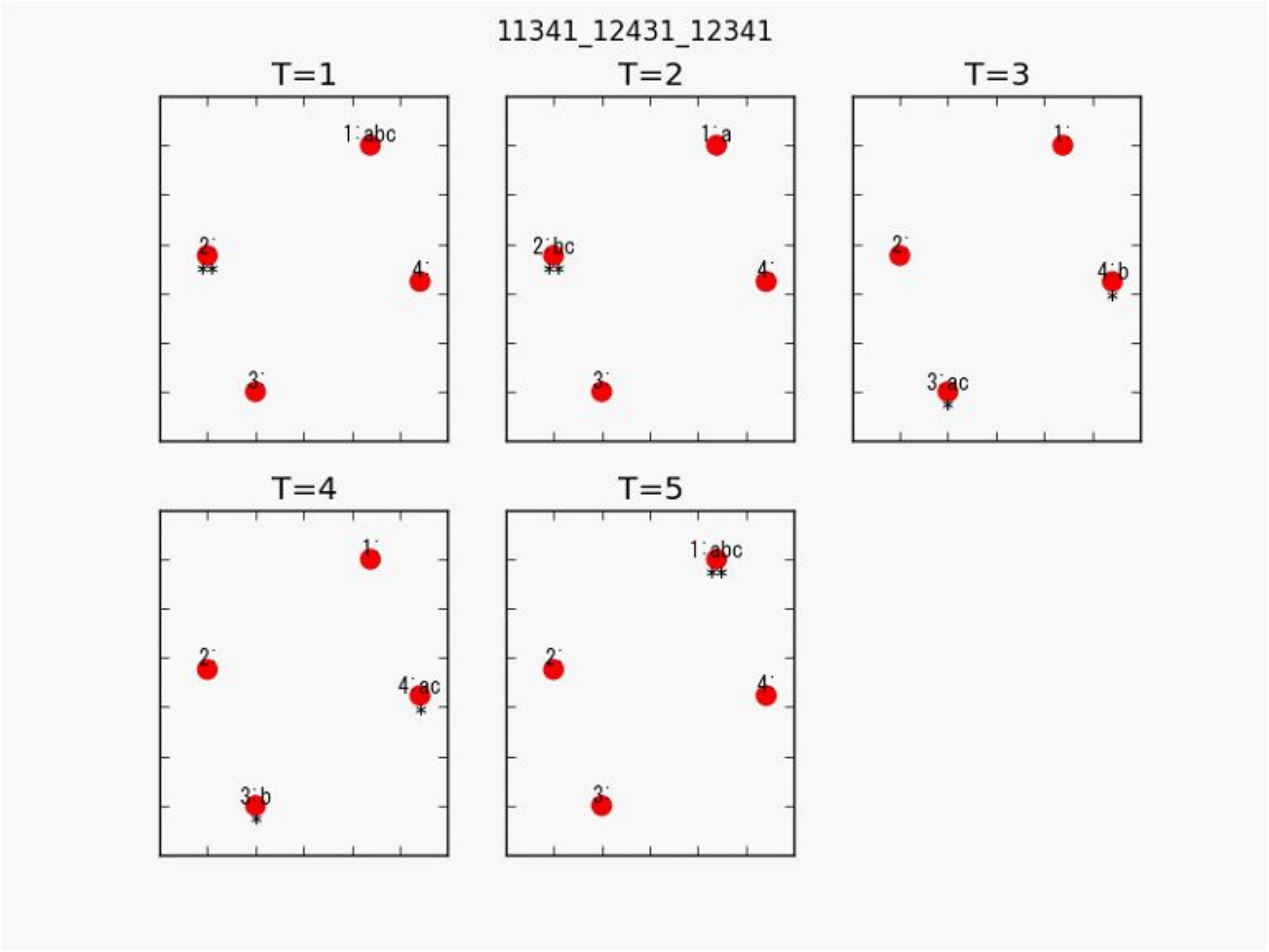}
	\caption{A player's restarting update}
	\label{fig:nonpg3}	
\end{figure}

\subsection{A case of FIP game}
\label{sec:pg}

\begin{hdn}
	\color{red}
	\begin{verbatim}
	\end{verbatim}
	\color{black}
\end{hdn}

Here, we assume $M=1$ and $w=4$ and that $G$ satisfies all the assumptions in Theorem \ref{thm:rsone} and has an FIP. In this case, the game immediately converges into a pNE as in Figure \ref{fig:pg}. The best update of player $a$ is to pickup the vehicle and all other players because the vehicle has enough capacity.

\begin{figure}[ht]
	\centering
	\includegraphics[clip,width=\fsize]{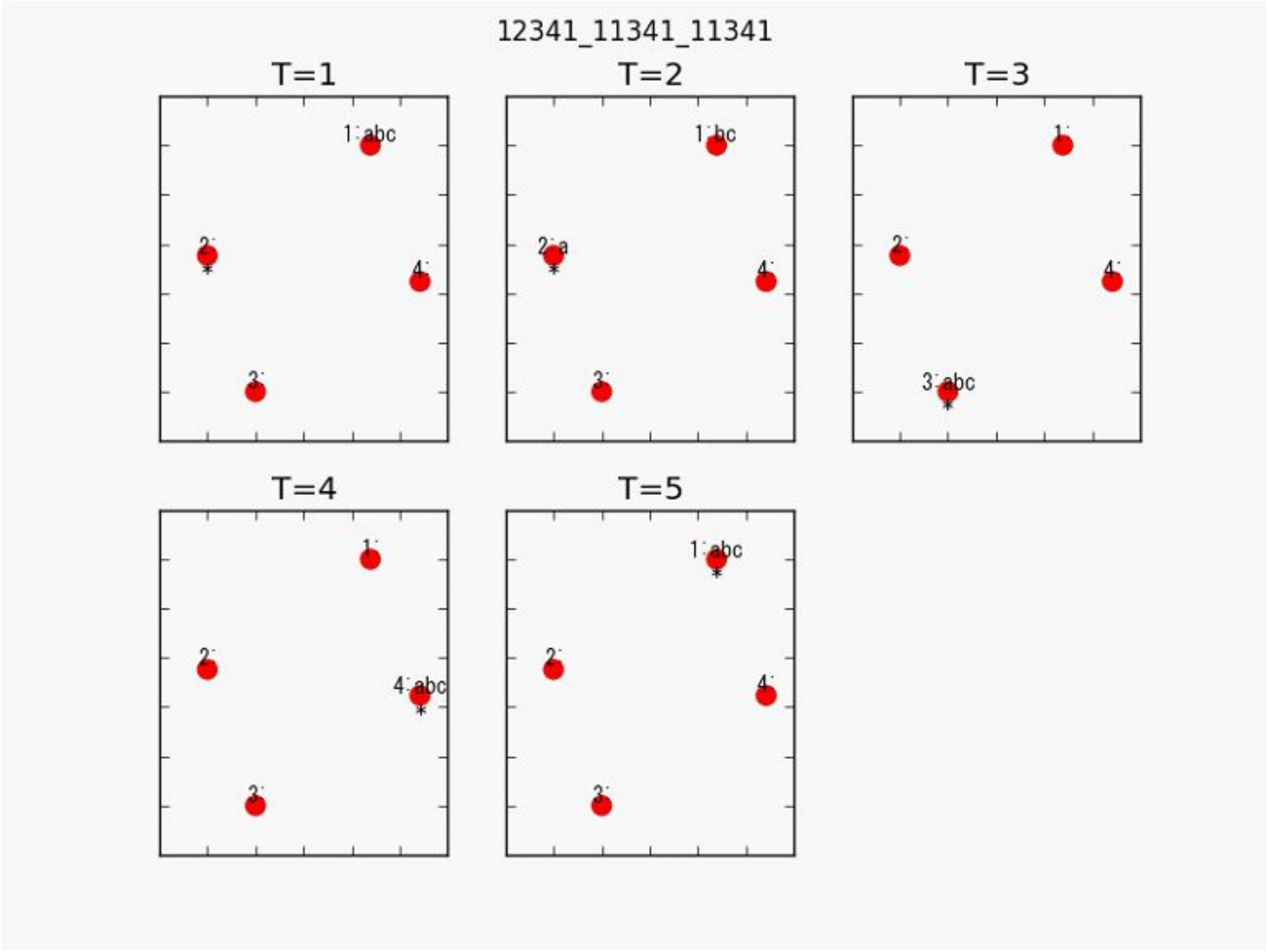}
	\caption{A pNE of the FIP game}
	\label{fig:pg}	
\end{figure}

\subsection{Another case of FIP game}
\label{sec:pg2}

Here, we consider another case where $M=2$ and $w=4$ and the initial profile is the same as that in Figure \ref{fig:nonpg1}. In this case, $a$ updates the same strategy as that in Figure \ref{fig:nonpg2} and increases the cost of player $b$. However, in this case, $b$ does not have to pick up the other vehicle but can be a passenger as in Figure \ref{fig:pg1}, and this is the same pNE as that in Figure \ref{fig:pg}. While this game does not satisfy H2 in Theorem \ref{thm:rsone}, it has an FIP. This means that the assumptions in Theorem \ref{thm:rsone} are not necessary conditions.

 \begin{figure}[ht]
 	\centering
 	\includegraphics[clip,width=\fsize]{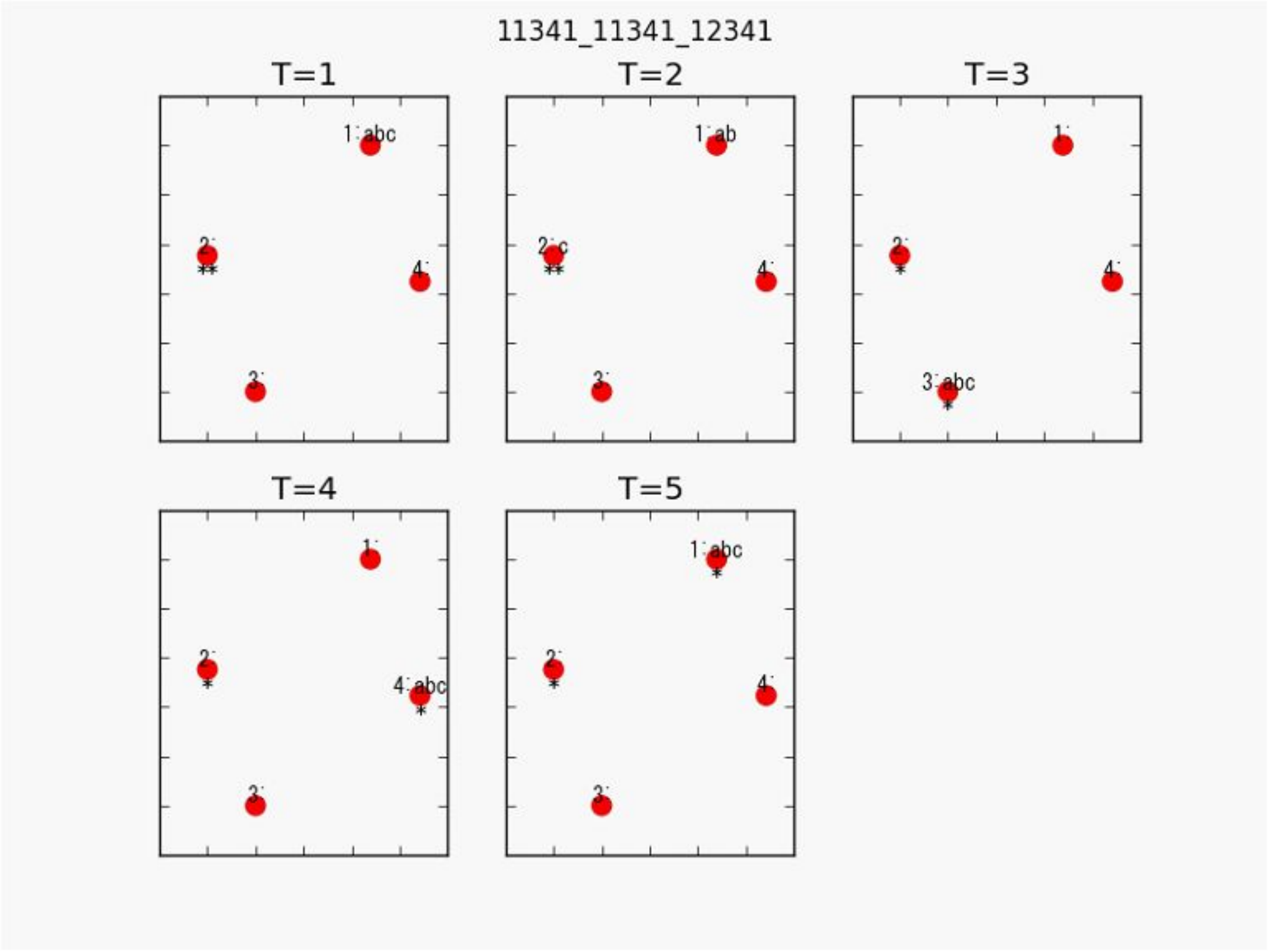}
 	\caption{Same pNE of the FIP game}
 	\label{fig:pg1}	
 \end{figure}
 
 \clearpage

\subsection{An example of signaling in ride sharing games}
\label{sec:bceex}

\begin{hdn}
	\color{red}
	\begin{verbatim}
PoAとかpgとかと関係ないですがー、今回は例による分析に留める、理解のためにー、と断る？
poa(bne) > poa(pne)=1
よって、シグナリングで改善を目指す
	\end{verbatim}
	\color{black}
\end{hdn}

Here, we show how signaling can improve the efficiency of sharing by giving players an incentive to coordinate with each other in a Bayesian ride sharing game. Game $G$ is defined as follows.

\begin{itemize}
	\item $N=2, V=3, T=4, M \le 1$.
	\item $\mc{G}$ and initial locations are shown in Figure \ref{fig:brsg}. All nodes have loop edges to themselves.
	\item $a_{i}$ must include node 3 for all players.
	\item $\mu$ is the first-fit linear allocation.
	\item There is an uncertainty $x \in \mc{X}=\{0,1\}$ regarding the existence of the vehicle. $x=0$ means $M=0$ and $x=1$ means $M=1$.
	\item All players have a common prior $p_{i}(x=0)=0.5, \forall i$. 
\end{itemize}

For each $x$, this game satisfies the assumptions in Theorem \ref{thm:rsone}. There are only two distinct options for each player that $\mc{A}=\{C,D\}$. $C=(1,2,3,1)$ is a trip that visits nodes in this order. On the other hand, $D=(1,1,3,1)$. All edges except for loop edges have the same cost function. If a player does not use the vehicle, the cost is 8. If a player drives alone, the cost is 6. If two players share the vehicle, the cost is 1. The cost of loop edges is zero. The cost matrices of this game are shown in Tables \ref{tbl:cost0} and \ref{tbl:cost1}.

 \begin{figure}[h]
 	\centering
 	\includegraphics[clip,width=6cm]{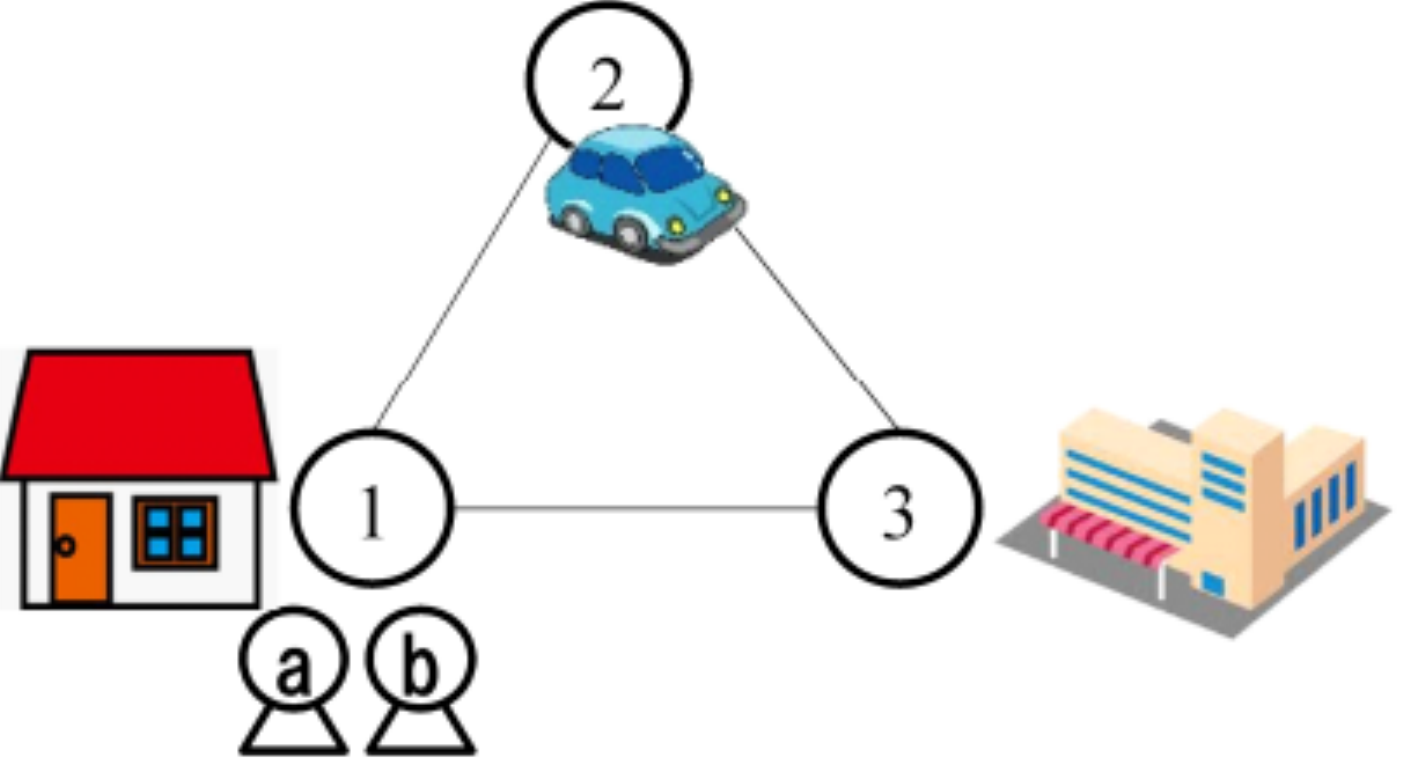}
 	\caption{A Bayesian ride sharing game}
 	\label{fig:brsg}	
 \end{figure}

The expected cost matrix is shown in Table \ref{tbl:ecost}. This matrix has the structure of a prisoner's dilemma and pBNE is $a=(D,D)$, which means no players share the vehicle.

Now we consider a system to coordinate players to share the unused vehicle by the BCE, as described in Section \ref{sec:signal}. A system cost can be denoted as $c_{s}(a|x)=\sum_{i \in \mc{N}}c_{i}(a|x)$. Then, the problem of the system is denoted as Eq.\ref{eq:problem}, which is the search for an optimal recommendation policy $\sigma(\hat{\vect{a}}|x)$ as in Table \ref{tbl:sigma}. The problem becomes a linear programming and Table \ref{tbl:sigmaopt} presents a solution. This incentive compatible recommendation induces the coordination of players as a BCE, where $\mathbb{E}_{x}[c_{s}(\hat{\vect{a}},x)]=27.9$, which is better than the one of pBNE $\mathbb{E}_{x}[c_{s}(\vect{a}|x)]=32$. Since $\mathbb{E}_{x}[c_{s}(\vect{a}|x)]=26$ in social optimum, the PoA is improved from 1.23 of pBNE to 1.07 of BCE.

\begin{table}[!]
	\centering
	\caption{$c_{i}(a_{1},a_{2}|x=0)$}
	\begin{tabular}{|c|c|c|} \hline
		& C & D \\ \hline \hline
		C & 20,20 & 20,16 \\ \hline
		D & 16,20 & 16,16 \\ \hline
	\end{tabular}
	\label{tbl:cost0}
	
	\caption{$c_{i}(a_{1},a_{2}|x=1)$}
	\begin{tabular}{|c|c|c|} \hline
		& C & D \\ \hline \hline
		C & 10,10 & 15,9 \\ \hline
		D & 9,15 & 16,16 \\ \hline
	\end{tabular}
	\label{tbl:cost1}
\end{table}

\begin{table}[!]
	\centering
	\caption{$\mathbb{E}_{x}[c_{i}(a_{1},a_{2}|x)]$}
	\begin{tabular}{|c|c|c|} \hline
		& C & D \\ \hline \hline
		C & 15,15 & 17.5,12.5 \\ \hline
		D & 12.5,17.5 & 16,16 \\ \hline
	\end{tabular}
	\label{tbl:ecost}
\end{table}


\begin{table}[!]
	\centering
	\caption{$\sigma(\hat{a}_{1},\hat{a}_{2}|x)$}
	\begin{tabular}{|c|c|c||c|c|c|} \hline
		\multicolumn{3}{|c||}{$\sigma(\hat{a}_{1},\hat{a}_{2}|0)$}  & \multicolumn{3}{|c|}{$\sigma(\hat{a}_{1},\hat{a}_{2}|1)$} \\ \hline \hline
		& C & D &  & C & D \\ \hline
		C & $\alpha_{0}$ & $\beta_{0}$ & C & $\alpha_{1}$ & $\beta_{1}$ \\ \hline
		D & $\beta_{0}$ & $1-\alpha_{0}-2\beta_{0}$ & D & $\beta_{1}$ & $1-\alpha_{1}-2\beta_{1}$ \\ \hline
	\end{tabular}
	\label{tbl:sigma}
\end{table}

\begin{table}[!]
	\centering
	\caption{Optimal $\sigma(\hat{a}_{1},\hat{a}_{2}|x)$}
	\begin{tabular}{|c|c|c||c|c|c|} \hline
		\multicolumn{3}{|c||}{$\sigma(\hat{a}_{1},\hat{a}_{2}|0)$}  & \multicolumn{3}{|c|}{$\sigma(\hat{a}_{1},\hat{a}_{2}|1)$} \\ \hline \hline
		& C & D &  & C & D \\ \hline
		C & 0 & 0 & C & 0.06 & 0.47 \\ \hline
		D & 0 & 1 & D & 0.47 & 0 \\ \hline
	\end{tabular}
	\label{tbl:sigmaopt}
\end{table}

\end{document}